\documentclass[conference]{IEEEtran}\IEEEoverridecommandlockouts
\usepackage{amsthm}
\usepackage{amssymb}
\usepackage{amsmath}
\usepackage{amsfonts}
\usepackage{graphicx}
\usepackage{algorithm}
\usepackage{algorithmic}
\usepackage{epstopdf}
\usepackage{cite}
\usepackage{hyperref}
\usepackage{subfigure}
\usepackage{mathrsfs}
\usepackage{float}
\usepackage{bm}
\newtheorem{theorem}{\textbf{Theorem}}

\newtheorem{lemma}{\textbf{Lemma}}

\newtheorem{remark}{\textbf{Remark}}

\begin{document}

\title{Artificial Interference Aided Physical Layer Security in Cache-enabled Heterogeneous Networks}
\author{Wu Zhao, Zhiyong Chen, Kuikui Li, and Bin Xia\\
Cooperative Medianet Innovation Center,  Shanghai Jiao Tong University,  Shanghai,  P. R. China\\
Email: \{zhaowu,zhiyongchen,kuikuili,bxia\}@sjtu.edu.cn}
\maketitle
\begin{abstract}
Caching popular contents is a promising way to offload the mobile data traffic in wireless networks, but so far the potential advantage of caching in improving physical layer security (PLS) is rarely considered. In this paper, we contribute to the design and theoretical understanding of exploiting the caching ability of users to improve the PLS in a wireless
heterogeneous network (HetNet). In such network, the base station (BS) ensures the secrecy of communication by utilizing some of the available power to transmit a pre-cached file, such that only the eavesdropper's channel is degraded. Accordingly, the node locations of BSs, users and eavesdroppers are first modeled as mutually independent poisson point processes (PPPs) and the corresponding file access protocol is developed. We then derive analytical expressions of two metrics, average secrecy rate and secrecy coverage probability, for the proposed system. Numerical results are provided to show the significant security advantages of the proposed network and to characterize the impact of network resource on the secrecy metrics.
\end{abstract}
\section{Introduction}
With the thriving development of mobile paying and internet of things, the privacy and security of wireless communication networks have become one of the most important issues. However, the broadcast nature of wireless channel leads to severe security vulnerabilities such as eavesdropping and jamming \cite{shiu2011physical}. To overcome these shortages, physical layer security (PLS) has emerged as a promising technology to complement and augment the security of wireless networks.

In \cite{wyner1975wire}, Wyner shows that when the eavesdropping channel is degraded than the main legitimating channel, the secrecy of communication can be perfectly guaranteed at a non-zero rate. And first, characterizes the maximal achievable secrecy rate as `\emph{secrecy capacity}' of the discrete wiretap channel.
Further, various efficient approaches are proposed to improve the secrecy capacity, e.g., artificial noise adding \cite{goel2008guaranteeing}, and relay cooperating\cite{dong2010improving}.
By exploiting multi-input single-output techniques, \cite{goel2008guaranteeing} proposes an artificial noise assisted beamforming scheme, which imposes the artificial noise into the null space of the legitimating channel to degrade the eavesdropping channel.
One source-destination pair with multiple relays intercepted by multiple eavesdroppers (ERs) is considered in \cite{dong2010improving}. By determining the relay weights, the authors maximize the achievable secrecy rate under different cooperating schemes.

With the popularization of and explosion of small communication equipments, the topology of the wireless network is becoming densely and randomly, which intensifies the concern for secure transmission.
Based on poisson point process (PPP) \cite{andrews2011tractable}, \cite{zhang2013enhancing,wang2013physical,deng2016physical,tang2016jammer} propose various schemes to improve physical layer security in such wireless heterogeneous network (HetNet).
In \cite{zhang2013enhancing}, the authors consider two transmission strategies based on sectoring and beamforming with artificial noise aided and investigate the secrecy capacity of both schemes.
By exchanging the location information between BSs, \cite{wang2013physical} analysis the effect of node locations on the achievable secrecy rate.
In \cite{deng2016physical}, the authors develop a tractable framework to analysis the average secrecy rate in a three-tier sensor network consisting of sensors, access points and sinks.
\cite{tang2016jammer} confound ERs with jamming signal from friendly jammers and artificial noise from full-duplex user. By selecting the jammer selection threshold to maximize secrecy probability.

Recently, caching popular contents at base station (BSs) and users has been introduced as a promising technique to address the mobile data tsunami in wireless networks \cite{Mobile3C,Offloading}.
The authors in \cite{sengupta2015fundamental} propose centralized and decentralized caching algorithms to guarantee secret transmission rate by coded multicast delivery.
Further, \cite{yang2017interference} utilizes the cached files of users as side information to cancel received interference from. However, the potential of caching in improving physical layer security is rarely considered until recently \cite{xiang2016cache}. The authors study the secure cooperative transmission among multiple cache-enabled BSs with the shared video data. Under the secrecy rate constraint, the total transmit power is minimized by jointly optimizing caching and transmission policies.

In this paper, we propose a heuristic scheme to enhance physical layer security in a cache-enabled HetNet by exploiting the caching ability of users. Instead of sending Gaussian noise as \cite{goel2008guaranteeing,zhang2013enhancing}, the BS transmits the target message combined with an artificial interference which is a file pre-cached at user. Since the cached file is known perfectly by the user, this part of interference can be erased as \cite{yang2017interference} while \cite{goel2008guaranteeing} and \cite{zhang2013enhancing} need the orthogonal space of legitimating channel to isolate noise. Meanwhile, ERs are confused by this part of artificial interference due to absence of this file. Specifically, by using stochastic geometry, we model the node locations (BSs, users and ERs) of the three-tier HetNet as mutually independent PPPs. The file access protocol is then proposed based on whether the file is cached or not and whether the user has cache ability or not. Accordingly, we derive analytical expressions of average secrecy rate and secrecy coverage probability for the proposed system in different transmission schemes. Numerical results show that proposed scheme can achieve promising performance in the both ER-dense and ER-sparse scenarios.

\section{System Model}
\subsection{Network Structure}
\begin{figure}[t]
\centering
\includegraphics[width=3.3in, height=1.5in]{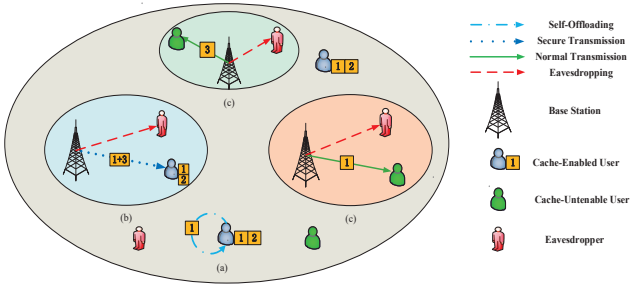}
\caption{System model of wireless cache-enabled heterogeneous networks with eavesdroppers, where $(a),(b)$ and $(c)$ stand of Self-offloading, Secure-transmission and Normal-transmission, respectively.}\label{System}
\end{figure}
As shown in Fig. \ref{System}, we consider a general wireless cache-enabled HetNet consisting three tiers of BSs, users and ERs, where the locations of BSs, users and ERs in each tier are spatially distributed based on independent PPPs, denoted as $\Phi\!_{b}$, $\Phi_{u}$ and $\Phi\!_{e}$ with density $\lambda_{b}$, $\lambda_{u}$ and $\lambda_{e}$, respectively. All nodes operate in single-antenna and we consider the downlink transmission, where time is divided into discrete slots with equal duration and we study one slot of the system. Large-scale fading and small-scale fading are both considered. We use $d^{-\beta}$ to denote large scale fading along the distance $d$, where $2\leqslant\beta\leqslant4$ is the path-loss exponent. For the small-scale fading, we consider the Rayleigh fading channel $h$, i.e., $\mid\!h\!\mid^{2}\sim\exp(1)$.

Consider a file library consisting of $N$ files denoted by $\mathcal{F}\!\triangleq\!\{f_{1},f_{2},\ldots,f_{N}\}$, and all the files are assumed to have equal length $L$. Each user randomly requests a file $f_{i}$ with probability $p_{i}$ and $\sum_{i=1}^{N}{p_{i}}=1$. Without loss of generality (w.l.o.g), we assume $p_1\geq p_2\geq \cdots\geq p_N$. Here, we also consider a file can only be stored entirely rather than partially. Non-colluding ERs intercept information by passive listening signal from BS.

We assume only $\alpha$ part of users have cache ability, where $0\leq\alpha\leq1$. The cache-enabled users also follow a thinning PPP with density $\alpha\lambda_{u}$. The cache-enabled users have same caching size with $(M\times L)bits$, where $M<N$ and cache the same $M$-most popular files out of $\mathcal{F}$ in this paper, which are marked as set $\mathcal{M}\triangleq\{f_{1},f_{2},\ldots,f_{M}\}$. To the aim of tractability, we assume that BSs can access all the files in $\mathcal{F}$ by directly connect to the core-network and neglect the extra cost for BS to fetch files. The set $\mathcal{M}$ will be broadcasted to all users by BSs at off-peak time then pre-stored at cache-enabled users. ERs considered in this paper have no cache ability.
\subsection{File Access Protocol}
Let $\mathcal{Q}$ be the total amount of request from users in $\Phi_{u}$ at one slot. As indicated in Fig. \ref{System}, the file access protocol can be described as follows:

\begin{enumerate}
 \item[(a)]\textbf{Self-offloading:} When a cache-enabled user happens to request a file in $\mathcal{M}$, the request will be satisfied and offloaded immediately from the user's local storage, termed as ``Self-offloading". By denoting the cache hit probability of the request fall in $\mathcal{M}$ as $\delta=\sum_{i=1}^{M} p_{i}$, the amount of this request is $\mathcal{Q}_{SO}={\alpha\delta\mathcal{Q}}$.
 \item[(b)]\textbf{Secure-transmission:} When a cache-enabled user requests $f_{i}$ ($i>M$) from the complementary set $\mathcal{F}
\!\setminus\!\mathcal{M}$ which is denoted as $\mathcal{C}\triangleq\{f_{M+1},f_{M+2},\ldots,f_{N}\}$, the target file $f_{i}$ can be provided by the nearest BS. In order to improve the transmission security, BS can combine the target file $f_i$ with a cached file $f_{m}$, i.e., $f_m\in\mathcal{M}$.\footnote{W.l.o.g, we use $f_1$ as $f_m$ in this paper which is noticed to all cache-enabled users.} Therefore, the transmission signal is $t_{i}\!=\!\sqrt{\theta\!P}x_{i}\!+\!\sqrt{(1\!-\!\theta)P}x_{m}$, where $P$ is the transmission power of BS, $x_{i}$ and $x_{m}$ are the signal of $f_{i}$ and $f_{m}$ with $E(|x_i|^2)=E(|x_m|^2)=1$ respectively, and $\theta\in(0,1]$ is the ratio of power allocation.

Since the pre-cached signal $x_m$ is known perfectly by the cache-enabled user, the user can cancel the extra interference $x_m$. However, $x_m$ is unknown for the ERs and thus can be viewed as a interference. We term this transmission as ``Secure-transmission". And the amount of this request is  $\mathcal{Q}_{ST}=\alpha(1-\delta)\mathcal{Q}$
 \item[(c)]\textbf{Normal-transmission:} When a user does not have
cache ability and requests $f_{i}$ from $\mathcal{F}$, $f_{i}$ will be transmitted by its nearest BS, termed as ``Normal-transmission". The transmission signal is $t_{i}=\sqrt{P}x_{i}$. Moreover, according to which subset of $f_i$ belongs to, the request can be divided into two types: $f_{i}\!\in\!\mathcal{M}$ and $f_{i}\!\in\!\mathcal{C}$.  Therefore the amount of these two types request are $\mathcal{Q}_{NT_{\mathcal{M}}}\!=\!{(1\!-\!\alpha)\delta\mathcal{Q}}$,  $\mathcal{Q}_{NT_{\mathcal{C}}}\!=\!{(1-\!\alpha)(1-\delta)\mathcal{Q}}$, respectively.

\end{enumerate}

In this paper, we assume all the BSs work in the full loaded state due to $\lambda_{u}\gg\lambda_{b}$ and each BS randomly serve one of user requests with equal probability. Therefore, the locations of BS in different states $\{(b),(c)\}$ are distributed as thinning PPPs $\Phi\!_{b_{1}},\Phi\!_{b_{2}},\Phi\!_{b_{3}}$ with density $\lambda_{b_{1}}\!=\!\frac{\mathcal{Q}_{ST}}{\mathcal{Q}_{ST}+\mathcal{Q}_{NT}}\lambda_{b}\!=\!\frac{\alpha(1-\delta)}{1-\alpha\delta}\lambda_{b}$,
$\lambda_{b_{2}}\!=\!\frac{\mathcal{Q}_{NT_\mathcal{M}}}{\mathcal{Q}_{ST}+\mathcal{Q}_{NT}}\lambda_{b}\!=\!\frac{(1-\alpha)\delta}{1-\alpha\delta}\lambda_{b}$ and $\lambda_{b_{3}}\!=\!\frac{\mathcal{Q}_{NT_\mathcal{C}}}{\mathcal{Q}_{ST}+\mathcal{Q}_{NT}}\lambda_{b}\!=\!\frac{(1-\alpha)(1-\delta)}{1-\alpha\delta}\lambda_{b}$ respectively.
\section{Analysis of Transmission Protocol}
According to Slivnyak's theorem \cite{elsawy2013stochastic}, a typical user $u_{0}$ locating at the origin of the Euclidean area does not change the distribution of PPP, no matter with or without caching ability. We also consider that the link between $u_{0}$ and its serving BS $b_{0}$ can be eavesdropped by all ERs in the network.
\subsection{Normal Transmission}
A typical user with no cache ability denoted as $un_{0}$ requests $f_{i}$ from $\mathcal{F}$. The nearest BS $b_{0}$ serves this request within the normal-transmission. Since the interference signals transmitted by other BSs from $\Phi_{b_1},\Phi_{b_2},\Phi_{b_3}$ cannot be cancelled without cached files, the interference $un_{0}$ suffering is equivalent coming from $\{\Phi_{b}\backslash b_{0}\}$ with power $P$.
Therefore the received signal of $un_{0}$ is
\begin{equation}\label{RS_UN}
y\!_{u\!n_{0}}\!\!=\!\!\sqrt{\!P}d_{u\!n_{0},b_{0}}^{-\frac {\beta}{2}}\!h_{u\!n_{0},b_{0}}x_{i}
+\!\!\!\!\!\!\sum_{k\in\{\!\Phi\!_{b}\!\backslash b_{0}\!\}}\!\!\!\!\sqrt{\!P}d_{u\!n_{0},b_{k}}^{-\frac {\beta}{2}}\!h_{u\!n_{0},b_{k}}x_{k'}
\!+n_{0},
\end{equation}
where $d_{u\!n_{0},b_{0}}$ denotes the distance between $u\!n_{0}$ and $b_{0}$, $h_{u\!n_{0},b_{0}}$ ($h_{u\!n_{0},b_{k}}$) represents the Rayleigh fading channel between $u\!n_{0}$ and $b_{0}$ ($b_{k}$), $x_{i}$ ($x_{k'}$\footnote{Note that $x_{k'}$ include secure transmission and normal transmission from BSs in $\{\Phi_{b}\backslash b_0\}$.}) is the transmission signal of $b_{0}$ ($b_k$), and $n_{0}\sim\mathcal{CN}(0,\sigma^{2})$ denotes the additive white Gaussian noise (AWGN). W.l.o.g, the variance of AWGN noise $n_{i}$ is $\sigma^2$ for $i=0,1,2,3$ in this paper.

Therefore the received signal-to-interference-plus-noise ratio (SINR) at $u\!n_{0}$ is
\begin{equation}\label{SINR_UN}
S\!I\!N\!R_{u\!n_{0}}=
\frac{ Pd_{u\!n_{0},b_{0}}^{-\beta}\!\!\mid\!\!h_{u\!n_{0},b_{0}}\!\!\mid^{2}}{\sum\limits_{k\in\{\Phi\!_{b}\!\backslash b_{0}\}}\!\! Pd_{u\!n_{0},b_{k}}^{-\beta}\!\!\mid\!\!h_{u\!n_{0},b_{k}}\!\!\mid^{2}\!\!+\sigma^{2}}.
\end{equation}

For the ER of $un_{0}$, the received signal at an arbitrary ER $e_{j}\in\Phi_{e}$ is similarly given by:
\begin{equation}\label{RS_EN}
y_{e\!_{j}n}\!=\!\sqrt{\!P}d_{e_{j},b_{0}}^{-\frac {\beta}{2}}h_{e_{j}\!,b_{0}}x_{i}
+\!\!\!\!\sum_{k\in\{\Phi\!_{b}\!\backslash b_{0}\}}\!\!\!\!\sqrt{P}d_{e_{j},b_{k}}^{-\frac {\beta}{2}}h_{e_{j},b_{k}}x_{k'}
\!+n_{1}.
\end{equation}

Because $x_{i}$ is eavesdropped signal for $e_{j}$, the SINR of $e_{j}$ can be written as
\begin{equation}\label{SINR_EN}
S\!I\!N\!R_{e_{j}n}=
\frac{ Pd_{e_{j},b_{0}}^{-\beta}\!\!\mid\!\!h_{e_{j},b_{0}}\!\!\mid^{2}}{\sum\limits_{k\in\{\Phi\!_{b}\!\backslash b_{0}\}}\!\!Pd_{e_{j},b_{k}}^{-\beta}\!\!\mid\!\!h_{e_{j},b_{k}}\!\!\mid^{2}\!+\sigma^{2}}.
\end{equation}
\subsection{Secure Transmission}
A typical user with cache ability denoted as $uc_{0}$ requests $f_{i}$ from $\mathcal{C}$. The nearest BS $b_{0}$ will serve this request within the secure-transmission. The received signal at $uc_{0}$ is given by:
\begin{align}
&y\!_{u\!c_{0}}\!=\!\sqrt{\theta\!P}d_{u\!c_{0},b_{0}}^{-\frac {\beta}{2}}h_{u\!c_{0}\!,b_{0}}x_{i}+\!\sqrt{(1-\theta) P}d_{u\!c_{0},b_{0}}^{-\frac {\beta}{2}}h_{u\!c_{0},b_{0}}x_{m}\nonumber\\
&\!{\setlength\arraycolsep{0.3pt}+}\!\!\sum_{j\in\{\!\Phi\!_{b_1}\!\!\backslash\!b_{0}\!\}}\!\!\!\left\{\!\!\sqrt{\theta\!P}d_{u\!c_{0},b_{j}}^{-\frac {\beta}{2}}\!h_{u\!c_{0},b_{j}}x_{j}\!+\!\sqrt{\!(1\!-\!\theta)\!P}d_{u\!c_{0},b_{j}}^{-\frac {\beta}{2}}h_{u\!c_{0},b_{j}}x_m\!\right\}\nonumber\\
&\!{\setlength\arraycolsep{0.3pt}+}\!\!\sum_{k\in\Phi\!_{b_{2}}}\!\!\!\!\sqrt{\!P}d_{u\!c_{0},b_{k}}^{-\frac {\beta}{2}}h_{u\!c_{0},b_{k}}x_{k}
\!+\!\!\!\sum_{l\in \Phi\!_{b_{3}}}\!\!\sqrt{\!P}d_{u\!c_{0},b_{l}}^{-\frac {\beta}{2}}h_{u\!c_{0},b_{l}}x_{l}\!{\setlength\arraycolsep{0.5pt} }+n_{2}\label{RS_UC}.
\end{align}

As described in Section II-B, the pre-cached signal $x_m$ is known perfectly at $uc_0$. And assume that the perfect channel state information (CSI) is fully available at cache-enabled users. Therefore, the $(1\!-\!\theta)$ part of interference from $\Phi_{b_{1}}$ and fully interference from $\Phi_{b_{2}}$ can be cancelled \cite{yang2017interference}. The SINR of $uc_{0}$ is
\begin{equation}\label{SINR_UC}
S\!I\!N\!R_{u\!c_{0}}\!=\!\frac{\theta\!Pd_{u\!c_{0},b_{0}}^{-\beta}\!\!\mid\!\!h_{u\!c_{0},b_{0}}\!\!\mid^{2}}{\!\theta\!P\!\!\!\underbrace{\!\!\sum\limits_{j
\in\{\!\Phi\!_{b_{1}}\!\!\backslash b_{0}\!\}}\!\!\!\!d_{u\!c_{0},b_{j}}^{-\beta}\!\!\mid\!\!h_{u\!c_{0},b_{j}}\!\!\!\mid^{2}}_{I_{\Phi\!_{b_{1}}}}{\setlength\arraycolsep{0.3pt}\!+}
P\underbrace{\!\!\sum\limits_{l\in\Phi\!_{b_{3}}}\!\!d_{u\!c_{0},b_{l}}^{- \beta}\!\!\mid\!\!h_{u\!c_{0},b_{l}}\!\!\!\mid^{2}\!}_{I_{\Phi\!_{b_{3}}}}+\sigma^{2}}.
\end{equation}

For the ER of $uc_{0}$, the received signal $y_{e_{j}c}$ is given by
\begin{align}\label{RS_EC}
y_{e\!_{j}c}&\!=\!\sqrt{\!\theta\!P}d_{e\!_{j}\!,b_{0}}^{-\frac {\beta}{2}}h_{e\!_{j}\!,b_{0}}x_{i}\!+\!\sqrt{\!(1\!-\!\theta) \!P}d_{e\!_{j}\!,b_{0}}^{-\frac {\beta}{2}}h_{e\!_{j}\!,b_{0}}x_m\nonumber\\
&+\!\!\!\!\sum_{k\in\{\Phi\!_{b}\!\backslash b_{0}\!\}}\!\!\!\!\!\!\sqrt{\!P}d_{e\!_{j}\!,b_{k}}^{-\frac {\beta}{2}}h_{e\!_{j}\!,b_{k}}x_{k'}
\!+n_{3},
\end{align}

Thus the SINR of $e_{j}$ can be calculated as
\begin{equation}\label{SINR_EC}
S\!I\!N\!R_{e\!_{j}c}\!=\!
\frac{\theta\!P\!\mid\!\!h_{e\!_{j},b_{0}}\!\!\mid^{2}\!\!d_{e\!_{j},b_{0}}^{-\beta}}{P\!\!\underbrace{\!\!\sum\limits_{k\in\{\Phi\!_{b}\backslash b_{0}\}}\!\!\!\!\!d_{e\!_{j},b_{k}}^{-\beta}\!\!\mid\!\!h_{e\!_{j},b_{k}}\!\!\mid^{2}\!}_{I_{\Phi\!_{b}}}+(1\!-\!\theta)P\!d_{e\!_{j},b_{0}}^{-\beta}
\!\!\mid\!\!h_{e\!_{j},b_{0}}\!\!\mid^{2}\!+\sigma^{2}\!}.
\end{equation}

\begin{remark}
We can observe from $(\ref{SINR_EC})$ that the expression has the form of $\frac{\theta\!X}{C+(1-\theta)X}$, where $X\!=\!P\!\mid\!h_{e_{j},b_{0}}\!\!\mid^{2}\!d_{e_{j},b_{0}}^{-\beta}$ which is a function of variables $h_{e_{j},b_{0}}$ and $d_{e_{j},b_{0}}$, while $C\!=\!P\!I_{\Phi\!_{b}}\!+\!\sigma^{2}$ is not relevant. Therefore we have $S\!I\!N\!R_{e\!_{j}c}\leq \frac{\theta}{1-\theta}\triangleq\gamma_{th_{0}}$.
\end{remark}
\section{Security Metrics Analysis}
In this section, the secrecy performance of two transmission protocols are compared in terms of average secrecy rate and secrecy coverage probability.
\subsection{Average Secrecy Rate}
Consider a link between the user $u_{0}$ and serving BS $b_{0}$ being intercepted by $E\!R\in\Phi_{e}$. We focus on the most detrimental ER which has the highest receive SINR from $b_{0}$.

The instantaneous secrecy rate $\mathcal{C}$ is thus given as
\begin{equation}
\mathcal{C}\triangleq\lceil C_{u}-C_{e}\rceil^\dagger \label{Cs},
\end{equation}
where$\lceil x\rceil^\dagger\!=\max\{x,0\}$. $C_{u}$ and $C_{e}$ are, respectively, the instantaneous capacity of the user's ($u_0$) channel and the most detrimental ER's channel, which can be expressed uniformly as $C_{i}\!=\!\log_2(1\!+\!\gamma_{i})$, $i=u,e$. Here, $\gamma_{e}$ is the instantaneous received SINR of the most detrimental ER, which is given by
\begin{equation}\label{max_SINR_EC}
\gamma_{e}\!=\!\max\limits_{e_{j}\in\Phi_{e}}\{S\!I\!N\!R_{e\!_{j}}\}.
\end{equation}

The average secrecy rate is defined as
\begin{equation}
\overline{\mathcal{C}}\triangleq\!\!\int_0^{\infty}\!\!\!\!\!\int_0^{\infty}\!\!\lceil C_{u}\!-C_{e}\rceil^\dagger \,d\gamma_{u}\! \,d\gamma_{e},
\end{equation}
and can be rewritten as \cite{deng2016physical}
\begin{equation}\label{Capacity1}
\overline{\mathcal{C}}\!=\!\frac{1}{\ln2}\!\int_0^{\infty}\!\!\!\big[1\!-\!F_{\gamma_{u}}(\gamma_{th})\big]\frac{F_{\gamma_{e}}\!(\!\gamma_{th}\!)}{1+\gamma_{th}} \,d\gamma_{th},
\end{equation}
where $F_{\gamma_{u}}$ and ($F_{\gamma_{e}}$) are the cumulative probability functions (CDFs) of $\gamma_{u}$ and ($\gamma_{e}$), respectively. Therefore, the $\overline{\mathcal{C}}$ of two transmission protocols are given as follow.\\

\subsubsection{Secure Transmission}
\begin{lemma}
Let $\gamma_{uc}$ be the SINR of the typical user with cache ability, the CDF of $\gamma_{uc}$ can be calculated as
\begin{align}\label{CDF_UC}
F\!_{\gamma\!_{u\!c}}\!(\!\gamma_{th}\!)\!=\!1\!-\!2\pi\lambda_{b}\!\!\!\int_0^{\infty}\!\!\!\!\!\!x\!\exp\!\Big\{\!&\!\!-\!\pi x^2\!\big[\!\mathcal{Z}(\gamma_{th})\lambda_{b_{1}}\!\!+\!\mathcal{Z}(\frac{\gamma_{th}}{\theta})\lambda_{b_{3}}\!+\!\lambda_{b}\big]\nonumber\\
&\!\!-\!\frac{\sigma^2}{\theta\!P}\gamma_{th}x^{\beta}\Big\}\!\,dx,
\end{align}
where $\mathcal{Z}(\!\gamma_{th}\!)\!=\!\frac{2\gamma_{th}}{\beta-2}{}_{2}F_{1}[1,1\!-\!\frac{2}{\beta};2\!-\!\frac{2}{\beta};-\gamma_{th}]$, ${}_{2}F_{1}[\cdot]$ is the Gauss hypergeometric function.
\end{lemma}
\begin{proof}
Please refer to Appendix A.
\end{proof}

\begin{lemma}
Let $\gamma_{ec}$ be the SINR of the most detrimental ER of the typical cache-enabled user, the CDF of $\gamma_{ec}$ is written as
\begin{equation}\label{CDF_EC}
F_{\gamma_{ec}}(\gamma_{th})= \\
 \begin{cases}
 \widetilde{F}_{\gamma_{ec}}\!(\gamma_{th}) &0\leq\gamma_{th}\leq\gamma_{th_{0}}\\
 \qquad1   &\qquad\text{else}, \\
 \end{cases}
\end{equation}
where $\widetilde{F}_{\gamma_{ec}}\!(\gamma_{th})$ is
\begin{align}\label{CDF_EC_2}
&\exp\!\!\Bigg\{\!\!\!-\!2\pi\lambda_{e}\!\!\!\int_0^{\infty}\!\!\!\!\!x\!\exp\!\bigg\{\!\!\!-\!\pi\lambda_{b}\Gamma(1\!+\!\frac{2}{\beta})\Gamma(1\!-\!\frac{2}{\beta})
\!\!\left[\!\frac{\gamma_{th}x^\beta}{\theta\!-\!(1\!\!-\!\theta)\gamma_{th}}\!\right]\!\!^{\frac{2}{\beta}} \nonumber \\
&\hspace{30mm} \!-\!{\frac{\sigma^2}{P}}\!\Big[\frac{\gamma_{th}x^\beta}{[\theta-\!(1\!-\theta)\gamma_{th}]}\Big]\!\bigg\}\!\,dx\!\Bigg\}.
\end{align}
and $\Gamma[\cdot]$ is the Gamma function.
\begin{proof}
Please refer to Appendix B.
\end{proof}
\end{lemma}

\begin{theorem}
In the interference-limited scenario, the average secrecy rate for the secure transmission is given by
\begin{align}\label{C_UC}
\overline{\mathcal{C}}\!_{ST}&\!=\!\frac{1}{\ln2}\!\int_0^{\gamma_{t\!h\!_{0}}}\!\!\frac{\exp\!\!\left\{\!-\frac{\lambda_{e}}{\lambda_{b}}\Big/\Gamma(1\!\!+\!\!\frac{2}{\beta})
\Gamma(1\!\!-\!\!\frac{2}{\beta})[\frac{\gamma_{t\!h}}{\theta-(1-\theta)\gamma_{t\!h}}]\!^{\frac{2}{\beta}}\!\!\right\}}{(1+\gamma_{t\!h})
[\mathcal{Z}(\gamma_{t\!h})\frac{\lambda_{b_{1}}}{\lambda_{b}}\!+\!\mathcal{Z}(\frac{\gamma_{t\!h}}{\theta})\frac{\lambda_{b_{3}}}{\lambda_{b}}\!+\!1]} \,d\gamma_{t\!h} \nonumber\\
&+\frac{1}{\ln2}\!\int_{\gamma_{t\!h\!_{0}}}^{\infty}\!\!\frac{\,d\gamma_{t\!h}}{(1\!+\!\gamma_{t\!h})
[\mathcal{Z}\!(\!\gamma_{t\!h}\!)\!\frac{\lambda_{b_{1}}}{\lambda_{b}}\!+\!\mathcal{Z}\!(\!\frac{\gamma_{t\!h}}{\theta}\!)\!\frac{\lambda_{b_{3}}}{\lambda_{b}}\!+\!1]}.
\end{align}
\end{theorem}
\begin{proof}
By substituting $(\ref{CDF_UC})$ and $(\ref{CDF_EC})$ into $(\ref{Capacity1})$, it is easy to obtain this theorem.
\end{proof}
\subsubsection{Normal Transmission}
\begin{lemma}
Let $\gamma_{un}$ as the SINR of the typical user without cache ability. Similar to $(\ref{CDF_UC})$, the CDF of $\gamma_{un}$ is
\begin{align}\label{CDF_UN}
&F\!_{\gamma\!_{u\!n}}\!(\!\gamma_{th}\!)\!=\!1\!\!-\!2\pi\lambda_{b}\!\!\!\int_0^{\infty}\!\!\!\!\!\!x\!\exp\!\!\left[\!-\pi\!\lambda_{b} x^2\!(\mathcal{Z}\!(\gamma_{th})\!+\!\!1\!)\!-\!\frac{\sigma^2}{P}\gamma_{th}x^\beta\!\right]\!\!\,dx.
\end{align}
\end{lemma}

\begin{lemma}
Let $\gamma_{en}$ as the SINR of the most detrimental ER of the typical user without cache ability. Similar to $(\ref{CDF_EC})$, it is easy to obtain the CDF of $\gamma_{en}$, which is given by
\begin{align}\label{CDF_EN}
&F\!_{\gamma_{e\!n}}\!(\!\gamma_{th}\!)\!=\!\exp\!\!\Bigg\{\!\!\!-\!2\pi\lambda_{e}\!\!\!\int_0^{\infty}\!\!\!\!\!x\!\exp\!\bigg\{\!\!\!-
\!\pi\lambda_{b}\Gamma(\!1\!+\!\frac{2}{\beta}\!)\Gamma(\!1\!-\!\frac{2}{\beta}\!)
\!\!\left[\!\gamma_{th}x^\beta\right]\!^{\frac{2}{\beta}} \nonumber \\
&\qquad\qquad\qquad\qquad\qquad\qquad\ \ -\!{\frac{\sigma^2}{P}}\gamma_{th}x^\beta\bigg\}\!\!\,dx\Bigg\}.
\end{align}
\end{lemma}

\begin{theorem}
In the interference-limited scenario, the average secrecy rate for the normal transmission is derived as
\begin{align}\label{C_UN}
\overline{\mathcal{C}}\!_{N\!T}&{\setlength\arraycolsep{0.3pt}=}\frac{1}{\ln2}\!\int_0^{\infty}\!\frac{\exp\!\left\{\!-\frac{\lambda_{e}}{\lambda_{b}}\Big/\Gamma(1\!\!+\!\!\frac{2}{\beta})
\Gamma(1\!\!-\!\!\frac{2}{\beta})\gamma_{t\!h}\!^{\frac{2}{\beta}}\!\right\}}
{(1+\gamma_{t\!h})[\mathcal{Z}(\gamma_{t\!h})\!+\!1]} \,d\gamma_{t\!h}.
\end{align}
\end{theorem}

\begin{proof}
By substituting $(\ref{CDF_UN})$ and $(\ref{CDF_EN})$ into $(\ref{Capacity1})$, we obtain this theorem.
\end{proof}

By comparing $(\ref{C_UC})$ and $(\ref{C_UN})$, we can find that $\overline{\mathcal{C}}_{ST}$ and $\overline{\mathcal{C}}_{NT}$ are both dependent on $\lambda_{e}/\lambda_{b}$, while $\overline{\mathcal{C}}_{ST}$ is also dependent on the power allocation ratio $\theta$ and the ratio of BS in three different states $\lambda_{b_i}/\lambda_{b}$, $i=1,3$. Note that $\lambda_{b_i}/\lambda_{b}$, $i=1,3$, are related to the cache-user ratio $\alpha$ and the cache hit ratio $\delta$. Numerical results will be given in Section V to show the effects of these parameters.
\subsection{Secrecy Coverage Probability}
Let $R_s$ be a given secrecy rate threshold. The delivery is securely successful when the instantaneous secrecy rate $\mathcal{C}$ is larger than the threshold $R_s$. Thus, the secrecy coverage probability can be expressed as
\begin{align}\label{p_coverage}
\mathcal{P}\!&\triangleq\!P_{r}(\mathcal{C}>\!R_s)\!=\!P_{r}[\log_2(1\!\!+\!\gamma_{u})\!-\!\log_2(1\!+\!\gamma_{e})\!\!>\!\!R_s] \nonumber\\
&=\mathbb{E}_{\gamma_{u},\gamma_{e}}\!\left\{\textbf{1}_{[\gamma_{u}>(1+\gamma_{e})2^{R\!_s}-1]}\right\} \nonumber\\
&=\int_0^{\infty}\!\!\!\!\int_{(1+\gamma_{e}\!)2^{R\!_s}\!-1}^{\infty}f_{\gamma_{u}}\!(\gamma_{1})f_{\gamma_{e}}\!(\gamma_{2})\,d\gamma_{1}\!\,d\gamma_{2}\nonumber\\
&=\int_0^{\infty}\!\!\!\!f_{\gamma_{e}}\!(\gamma_{th})\!\big\{\!1\!-\!F_{\gamma_{u}}\![2^{R_s}\!(1\!+\!\gamma_{th}\!)\!-\!1]\big\}\! \,d\gamma_{th},
\end{align}
where $f_{\gamma_{u}}\!(\gamma_{u})$ and $f_{\gamma_{e}}\!(\gamma_{e})$ are the probability distribution functions (PDFs) of $\gamma_{u}$ and $\gamma_{e}$, respectively.
\begin{theorem}
In the interference-limited scenario with the secure transmission, the secrecy coverage probability is
\begin{align}\label{PST}
\mathcal{P}\!_{S\!T}&\!=\!\!\int_0^{\gamma_{t\!h\!_{0}}}\!\!\!\bigg\{\frac{\exp\!\!\big[\!-\frac{\lambda_{e}}{\lambda_{b}}\Big/\Gamma(1\!\!+\!\!\frac{2}{\beta})
\Gamma(1\!\!-\!\!\frac{2}{\beta})[\frac{\gamma_{t\!h}}{\theta-(1-\theta)\gamma_{t\!h}}]\!^{\frac{2}{\beta}}\!\big]}
{\mathcal{G}(\!{\gamma_{th}}\!)\frac{\lambda_{b_{1}}}{\lambda_b}\!+\!\mathcal{G}(\!\frac{\gamma_{th}}{\theta}\!)
\frac{\lambda_{b_{3}}}{\lambda_b}\!+\!1} \nonumber \\
&\quad\ \ \frac{2\lambda_{e}\theta\gamma_{th}^{-\frac{\beta+2}{\beta}}}{\beta\lambda_{b}\Gamma(1+\frac{2}{\beta})\Gamma(1-\frac{2}{\beta})
[\frac{\gamma_{t\!h}}{\theta-(1-\theta)\gamma_{t\!h}}]^{\frac{\beta-2}{\beta}}}\bigg\} \,d\gamma_{th},
\end{align}
where $\mathcal{G}(\!\gamma_{t\!h}\!)$ is given as
\begin{equation}
\mathcal{G}(\!\gamma_{t\!h}\!)\!=\!
{[(1\!+\!\gamma_{th}\!)2^{R_s}\!\!-\!1]}^{\frac{2}{\beta}}\!\!\!\int_{{[(1\!+\!\gamma_{th}\!)2^{R_s}\!\!-\!1]}^{\frac{2}{\beta}}}^
{\infty}\!\frac{1}{1\!+\!x^{\frac{\beta}{2}}}dx.
\end{equation}
\end{theorem}
\begin{proof}
By differentiating $(\ref{CDF_EC})$ to get $f_{\gamma_{e}}\!(\gamma_{e})$, then substituting $f_{\gamma_{e}}\!(\gamma_{e})$ and $(\ref{CDF_UC})$ into $(\ref{p_coverage})$, we obtain this theorem.
\end{proof}

\begin{theorem}
In the interference-limited scenario with the normal transmission, the secrecy coverage probability is
\begin{align}\label{PNT}
\mathcal{P}\!_{N\!T}&\!=\!\!\int_0^{\infty}\!\!\!\bigg\{\!\frac{\exp\!\big[\!-\!\frac{\lambda_{e}}{\lambda_{b}}\Big/\Gamma(1\!\!+\!\!\frac{2}{\beta})
\Gamma(1\!\!-\!\!\frac{2}{\beta})\gamma_{t\!h}\!^{\frac{2}{\beta}}\!\big]}
{\mathcal{G}({\gamma_{th}})+1} \nonumber \\
&\qquad\quad\!\frac{2\lambda_{e}\gamma_{th}^{-\frac{\beta+2}{\beta}}}{\beta\lambda_{b}
\Gamma(\!1+\frac{2}{\beta}\!)\Gamma(1-\frac{2}{\beta})
}\bigg\} \,d\gamma_{th}
\end{align}
\end{theorem}
\begin{proof}
By differentiating $(\ref{CDF_EN})$ to get $f_{\gamma_{e}}\!(\gamma_{e})$, then substituting $f_{\gamma_{e}}\!(\gamma_{e})$ and $(\ref{CDF_UN})$ into $(\ref{p_coverage})$, we get this theorem.
\end{proof}
\section{Numerical Results}
\begin{figure}[t]
\centering
\includegraphics[width=3.1in,height=2.2in]{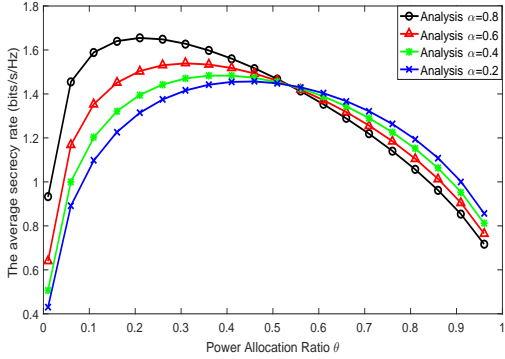}
\vspace{-2mm}
\caption{ The average secrecy rate $\overline{\mathcal{C}}_{ST}$ v.s. power allocation ratio $\theta$}\label{c_p}
\label{optimalpower}
\vspace{1.5mm}
\end{figure}
\begin{figure}[t]
\centering
\includegraphics[width=3.1in,height=2.2in]{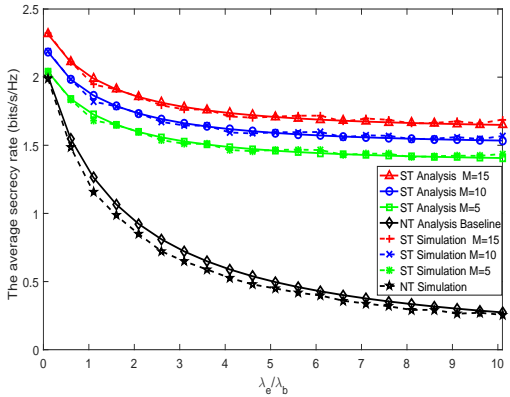}
\vspace{-2mm}
\caption{ The average secrecy rate $\overline{\mathcal{C}}$ v.s. $\lambda_{e}/\lambda_{b}$ }\label{capacity}
\label{capacity}
\vspace{1.5mm}
\end{figure}
In this section, numerical results are provided to evaluate the performance of the proposed transmission schemes. The BSs, ERs and users are distributed based on PPPs with density $\{\lambda_{b},\lambda_{e},\lambda_{u}\}=\{1,5,100\}/km^2$ in the simulation. We consider the transmission power $P=30~dBm$ and the noise power $\sigma^2=-174~dBm$. We consider the path loss exponent $\beta=4$, the total number of files $N=100$, the cache size $M=5$, the power allocation ratio $\theta=0.5$, and the cache user ratio $\alpha=0.5$. In the simulation, the file popularity distribution is modeled as Zipf distribution, i.e., the requested probability of the $i$-th ranked file is given by $p_{i}=\frac{1/{i^{\eta}}}{\sum_{j=1}^{N}{1/j^{\eta}}}$ where $\eta\geq0$ characterizes the skew of the popularity distribution. We use $\eta=0.8$ in the simulation. These parameters will not change unless specified otherwise.

In Fig.\ref{optimalpower}, the average secrecy rate of the secure transmission $\overline{\mathcal{C}}_{ST}$ versus the power allocation ratio $\theta$ is illustrated. It can be seen that there exists an optimal $\theta^*$ to achieve the maximal $\overline{\mathcal{C}}\!_{ST}$ for a given $\alpha$, and different $\alpha$ has different $\theta^*$. As presented in (\ref{C_UC}), $\overline{\mathcal{C}}\!_{ST}$ cannot be expressed in a closed form. As such, we cannot derive $\theta^*$ in theory. We can observe that the average secrecy rate $\overline{\mathcal{C}}\!_{ST}$ first increase with $\theta$ when $\theta\in(0,\theta^*)$, then decreases with $\theta$ when $\theta\in(\theta^*,1)$. This interesting phenomenon can be well explained from $(\ref{SINR_UC})$ and $(\ref{SINR_EC})$. The increase of $\theta$ improves the SINRs of both user and ER, but the increment at user is dominant in $(0,\theta^*)$. When $\theta$ is getting larger, the $\overline{\mathcal{C}}\!_{ST}$ will be compromised due to the growing effects of eavesdropping. We can also obtain that the secure transmission can achieve better optimal $\overline{\mathcal{C}}$ in larger $\alpha$ scenario, because more secure transmissions occurs in the network.



\begin{figure}[t]
\centering
\includegraphics[width=3.1in,height=2.2in]{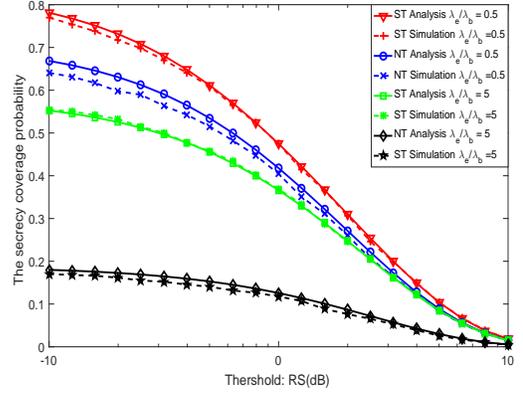}
\vspace{-2mm}
\caption{ The secrecy coverage probability vs. thershold}\label{c_p}
\label{pcoverage}
\vspace{1.5mm}
\end{figure}
In Fig.\ref{capacity}, the average secrecy rate $\overline{\mathcal{C}}$ versus the density ratio of $\frac{\lambda_{e}}{\lambda_{b}}$ with different cache size $M$ is illustrated. Note from (\ref{C_UN}) that $\overline{\mathcal{C}}_{NT}$ is only depend on $\frac{\lambda_{e}}{\lambda_{b}}$ which is considered as a baseline. We can see that with the increase of $\frac{\lambda_{e}}{\lambda_{b}}$, the $\overline{\mathcal{C}}$ is decreased for both with and without cache-enabled transmission schemes, which indicates that more ERs cause more serious eavesdropping. It should be highlighted that, even with $\frac{\lambda_{e}}{\lambda_{b}}\!=\!10$, the $\overline{\mathcal{C}}_{ST}$ is still above $1.5~bits/\!s/\!Hz$ which only reduced $25\%$ from above $2~bits/\!s/\!Hz$ when $\frac{\lambda_{e}}{\lambda_{b}}=0.1$, while $\overline{\mathcal{C}}\!_{N\!T}$ reduces to $0.3~bits/\!s/\!Hz$ from $2~bits/\!s/\!Hz$, i.e., reduced by $85\%$. In addition, we can also observe that the $\overline{\mathcal{C}}\!_{S\!T}$ improves with increasing cache size $M$, due to more interference signal can be cancelled with larger ratio of $\Phi_{b_1}$ and $\Phi_{b_2}$.

In Fig.\ref{pcoverage}, the secrecy coverage probability $\mathcal{P}$ for various $\frac{\lambda_{e}}{\lambda_{b}}$ is presented. We can observe that $\mathcal{P}_{ST}$ is much higher than $\mathcal{P}_{NT}$ for both $\frac{\lambda_{e}}{\lambda_{b}}=5$  and $0.5$, which indicates the promising effect of secure transmission in both ER-dense scenario and ER-sparse scenario. Similar as Fig.\ref{capacity}, we can also see that more ERs cause more serious eavesdropping leading lower secrecy coverage probability. The simulation results are presented along with the theoretical ones in Fig.\ref{capacity} and Fig. \ref{pcoverage}. We can see from the figures that the theoretical results are in excellent agreement with the simulation results.
\section{Conclusion}
In this paper, we reveal that the caching ability of users can be used to improve the transmission security for physical layer security in the wireless cache-enabled HetNet. The corresponding secure transmission scheme is developed, where the transmitter combines the message signal with the pre-cached file. This scheme can introduce extra interference at ER, but this interference can be cancelled at the cache-enabled users. Based on stochastic geometry, we derive the expression of average secrecy rate and secrecy coverage probability for the secure transmission and the normal transmission. Finally, we show that the secure transmission achieves a significant security gain than the normal transmission.
\begin{appendix}
\subsection{Proof of Lemma 1}
By replacing the distance $d_{u_{0},b_{0}}$ between the typical user and its nearest BS with $x$, the PDF of $x$ is $f_{X}(x)\!=\!2\!\pi\lambda_{b}x\exp\{-\pi\lambda_{b}x^2\}$\cite{andrews2011tractable}. Then the CDF of $\gamma_{uc}$ is derived as
\begin{align}\label{F_UC}
 &F_{\gamma_{uc}}(\gamma_{th})\!\triangleq\!\!P_{r}[S\!I\!N\!R_{uc}\leq\gamma_{th}]\!=\!\mathbb{E}_{x}[S\!I\!N\!R_{uc}\leq\gamma_{th}|d_{u_{0},b_{0}}\!=\!x]\nonumber\\
 &=\int_0^{+\infty}\!\!\!P_{r}\!\Big[\frac{\theta\!P\!\mid\! h_{u_{0},b_{0}}\!\mid^{2}\!x^{-\beta}}{\theta\!P\!I_{\Phi\!_{b_{1}}}\!\!+P\!I_{\Phi\!_{b_{3}}}\!\!+\!\sigma^{2}}\leq\gamma_{th}\Big]\!f_{X}(x)\,dx\nonumber\\
 &=\int_0^{+\infty}\!\!\!\!P_{r}\Big[\!\!\mid\! h_{u_{0},b_{0}}\!\mid^{2}\leq\!\gamma_{th}x^{\beta}(I_{\Phi\!_{b_{1}}}\!\!+\!\frac{I_{\Phi\!_{b_{3}}}}{\theta}\!+\!\frac{\sigma^{2}}{\theta\! P})\Big]\!f_{X}(x)\,dx\nonumber\\
 &\overset{(a)}{=}1\!{\setlength\arraycolsep{0.3pt}-}\!\!\int_0^{\infty}\!\!\!\!\mathcal{L}_{I_{\Phi\!_{b_{1}}}}\!\!(\gamma_{th}x^{\beta})\mathcal{L}_{I_{\Phi\!_{b_{3}}}}\!\!(\frac{\gamma_{th}x^{\beta}}{\theta})
 e^{-\frac{\gamma_{th}x^{\beta}\!\sigma^{2}}{\theta\!P}}\!f\!_{X}(x)\,dx,
\end{align}
where Step (a) follows from $\mid\! h_{u_{0},b_{0}}\!\mid^{2}\sim\exp(1)$. Under the condition of $d_{u_{0},b_{0}}\!=\!x$, the remaining interferences resulted from $\Phi\!_{b_{1}}$ and $\Phi\!_{b_{3}}$ are spatially located at the outside of the circle centered at $u_{0}$ with radius $x$ denoted as $\mathcal{C}_{(u_0,x)}$. Therefore the Laplace transform $\mathcal{L}_{I_{\Phi\!_{b_{1}}}}$ is derived as
\begin{align}\label{L_I_B1}
&\mathcal{L}_{I_{\Phi\!_{b_{1}}}}\![\gamma_{th}x^{\beta}]
=\mathbb{E}_{I_{\Phi\!_{b_{1}}}}\!\!\Big[\exp({-\gamma_{th}x^{\beta}\!\!\!\!\!\sum\limits_{j\in\{\Phi\!_{b\!_{1}}\!\!\backslash b_{0}\}}\!\!\!\!\!d_{u_{0},b_{j}}^{-\beta}\!\!\mid\!\!h_{u_{0},b_{j}}\!\!\mid^{2}})\Big]\nonumber\\
&=\mathbb{E}_{I_{\Phi\!_{b_{1}}}}\!\!\Big\{\!\prod_{j\in\{\Phi\!_{b_{1}}\!\!\backslash b_{0}\!\}}\!\!\!\!\!\!\big[\!\exp\!\big(\!-\!\gamma_{th}x^{\beta}d_{u_{0},b_{j}}^{-\beta}\!\!\mid\!h\!_{u_{0},b_{j}}\!\!\mid^{2})\big]\!\Big\}\nonumber\\
&\overset{(a)}{=}\!\exp\!\bigg\{\!\!-\!\lambda_{b_{1}}\!\!\!\int_{R^2\backslash\mathcal{C}_{(u_0,x)}}\!\!\!\!\Big[1\!-\!\mathbb{E}_{\mid h_{u_{0},b_{j}}\mid^{2}}(e^{-\gamma_{th}x^{\beta}r^{-\beta}})\Big]\!\,dr\!\bigg\}\nonumber\\
&\overset{(b)}{=}\exp\!\Big[\!-\!2\pi\lambda_{b_{1}}\!\!\int_{x}^{\infty}\!\!\frac{v}{1+(\gamma_{th}x^{\beta})^{-1}v^\beta}\,dv\Big]\nonumber\\
&\overset{(c)}{=}\exp\!\big[\!-\!\pi\lambda_{b_{1}}x^2\gamma_{th}^{\frac{2}{\beta}}\int_{\gamma_{th}^{-\frac{2}{\beta}}}^{\infty}
\frac{\,dy}{1+y^{\frac{\beta}{2}}}\big]\nonumber\\
&=\exp\!\big[\!-\!\pi\lambda_{b_{1}}x^2\mathcal{Z}(\gamma_{th})\big],
\end{align}
where step (a) follows from the probability generating functional (PGFL) of PPP, step (b) is obtained by converting the cartesian coordinates into polar coordinates, step (c) is obtained by replacing $(\gamma_{th}x^{\beta})^{-\frac{2}{\beta}}v^2$ with $y$.

Similarly, the Laplace transform of the $\mathcal{L}_{I_{\Phi\!_{b_{3}}}}$ is
\begin{align}\label{L_I_B3}
\mathcal{L}_{I_{\Phi\!_{b_{3}}}}\!\big[\frac{\gamma_{th}x^{\beta}}{\theta}\big]
&=\mathbb{E}_{I_{\Phi\!_{b_{3}}}}\!\!\Big[\exp({-\frac{\gamma_{th}x^{\beta}}{\theta}\!\!\sum\limits_{l\in\Phi\!_{b\!_{3}}\!\!}\!\!d_{u_{0},b_{l}}^
{-\beta}\!\!\mid\!\!h_{u_{0},b_{l}}\!\!\mid^{2}})\Big]\nonumber\\
&=\exp\!\big[\!-\!\pi\lambda_{b_{3}}x^2\mathcal{Z}(\frac{\gamma_{th}}{\theta})\big].
\end{align}

By substituting $(\ref{L_I_B1})$, $(\ref{L_I_B3})$ into $(\ref{F_UC})$, we can obtain Lemma 1 and the proof is completed.
\subsection{Proof of Lemma 2}
Let us replace the distance $d_{e_{i},b_{0}}$ with $x$. When $\gamma_{ec}\leq\gamma_{th_{0}}$, by substituting $(\ref{SINR_EC})$ into $(\ref{max_SINR_EC})$, the CDF of $\gamma_{ec}$ is
\begin{align}\label{F_EC}
&F\!_{\gamma_{ec}}\!(\gamma_{th})\triangleq P_{r}\big\{\max\limits_{e_{i}\in\Phi_{e}}[S\!I\!N\!R_{e_{i}}\leq\gamma_{th}]\big\}\nonumber\\
&=P_{r}\!\bigg\{\!\max\limits_{e_{i}\in\Phi_{e}}\!\Big[\frac{\theta P\!\mid\! \!h_{e_{i},b_{0}}\!\!\mid^{2}\!d_{e_{i},b_{0}}^{-\beta}}{P\!I_{\Phi\!_{b}}\!+\sigma^{2}\!+\!(1\!-\theta)P\!\mid\! h_{e_{i},b_{0}}\!\!\mid^{2}\!d_{e_{i},b_{0}}^{-\beta}}\leq\gamma_{th}\Big]\bigg\} \nonumber\\
&=\mathbb{E}_{\Phi\!_{e}}\!\bigg[\!\prod_{i\in\Phi\!_{e}}\!\!\!P_{r}\!\Big(\frac{\theta P\!\mid\! h_{e_{i},b_{0}}\!\!\mid^{2}\!d_{e_{i},b_{0}}^{-\beta}}{P\!I_{\Phi\!_{b}}\!+\sigma^{2}\!+\!(1\!-\!\theta)P\!\!\mid\!\! h_{e_{i},b_{0}}\!\!\mid^{2}\!d_{e_{i},b_{0}}^{-\beta}}\leq\gamma_{th}\Big)\!\bigg]\nonumber\\
&\overset{(a)}{=}\!\exp{\!\bigg\{\!\!\!-\!\lambda_{e}\!\!\!\int_{R^2}\!\!\!\!\!1\!\!-\!P_{r}\!\Big[\!\!\!\mid\! h_{e_{i},b_{0}}\!\!\mid^{2}\leq\!\frac{(\!P\!I_{\Phi\!_{b}}\!\!+\!\sigma^{2})\gamma_{th}r^\beta}
{[\theta-\!(1\!-\theta)\gamma_{th}]P}\!\Big]\!\,dr\!\!\bigg\}} \nonumber\\
&\overset{(b)}{=}\exp{\bigg\{\!\!-\!2\pi\lambda_{e}\!\!\!\int_0^{\infty}\!\!\!\!x\mathcal{L}_{I_{\Phi\!_{b}}}
\!\!\Big[\frac{\gamma_{th}x^{\ \beta}}{[\theta-\!(1\!-\theta)\gamma_{th}]}\Big]\!
\mathrm{e^{\!-\!\frac{\sigma^{2}\gamma_{th}\!x^{\beta}}{[\theta-(1-\theta)\gamma_{th}]\!P}}}\!\,dx\!\bigg\}} 
\end{align}
where step (a) follows from the PGFL of PPP, step (b) is obtained by converting the cartesian coordinates into polar coordinates.

Note that the interference comes from $\{\Phi\!_{b}\backslash b_{0}\}$, which is the reduced Palm distribution of PPP $\Phi\!_{b}$. According to Slivnyak-Mecke theorem \cite{chiu2013stochastic}, the reduced Palm distribution of PPP is equivalent of its original distribution, i.e.$\{\Phi\!_{b}\backslash b_{0}\}=\Phi\!_{b}$ as illustrated in \cite{tang2016jammer}. Denoting $S=\frac{\gamma_{th}x^\beta}{[\theta-(1-\theta)\gamma_{th}]}$, the Laplace transform of interference $I_{\Phi_{b}}(S)$ is derived as
\begin{align}\label{L_B}
&\mathcal{L}_{I_{\Phi\!_{b}}}\!(S)\!=\mathbb{E}_{I_{\Phi\!_{b}}}\big[\mathbf{e}^{-SI_{\Phi\!_{b}}}\big] \nonumber\\
&=\mathbb{E}_{I_{\Phi\!_{b}}}\!\Big[\exp{\!\big(\!-\!S\!\!\!\!\sum_{j\in\{\Phi\!_{b}
\backslash b_{0}\}}\!\!\!\!\mid\!\!h_{e_{i},b_{j}}\!\!\mid^{2}\!\!d_{e_{i},b_{j}}^{-\beta}\big)}\Big] \nonumber\\
&\overset{(a)}{\!=\!}\!\exp{\!\left\{\!\!-\!\lambda_{b}\!\!\int_{R^2}\!\!\!\!1\!\!-\!\mathbb{E}_{\mid h_{e_{i},b_{j}}\!\mid^2}\!\!\left[\exp{\!\big(\!\!-\!S\!\mid\! \!h_{e_{i},b_{j}}\!\!\mid^{2}\!\!d_{e_{i},b_{j}}^{-\beta}\big)}\!\right]\!\!\,d(d_{e_{i},b_{j}})\!\right\}} \nonumber\\
&\overset{(b)}{=}\exp\!{\Big\{\!\!-\!2\pi\lambda_{b}\!\!\int_0^{\infty}\!\!\frac{v}{1+\frac{v^\beta}{S}}\,dv\Big\}}\nonumber\\
&=\exp\!\Big[\!-\!\pi\lambda_{b}\Gamma(1\!+\!\frac{2}{\beta})\Gamma(1\!-\!\frac{2}{\beta})S^{\frac{2}{\beta}}\Big],
\end{align}
where step (a) follows from the PGFL of PPP,
where step (b) is obtained by converting cartesian coordinates into polar coordinates.

By substituting $(\ref{L_B})$ into $(\ref{F_EC})$, we can get $\widetilde{F}_{\gamma_{ec}}\!(\gamma_{th}) \text{as}~(\ref{CDF_EC_2})$. When $\gamma_{ec}$ is larger than $\gamma_{th_{0}}$, it is clearly to note that $F_{\gamma_{ec}}(\gamma_{th})=1$. Then the proof is completed.
\end{appendix}
\bibliographystyle{IEEEtran}
\bibliography{refe}

\end{document}